\documentclass[conference]{IEEEtran}
\usepackage{amsmath}
\usepackage{amsthm}
\usepackage{amsfonts}
\usepackage{color}
\usepackage{mathtools}
\usepackage{relsize}
\usepackage{scalerel}
\usepackage{bussproofs}
\usepackage{relsize}
\usepackage{tikz}
\usepackage{proof}
\usepackage{amssymb}
\usepackage{etex}
\usepackage[noadjust]{cite}
\usetikzlibrary{arrows,positioning}
\usetikzlibrary{backgrounds}

\newtheorem{definition}{Definition}
\newtheorem{example}{Example}

\newtheorem{corollary}{Corollary}
\newtheorem{theorem}{Theorem}

\newtheorem{lemma}{Lemma}

\newcommand{\Nat}{\ensuremath{\mathbb{N}}}
\newcommand{\Var}{\ensuremath{\mathcal{V}}}
\newcommand{\RecV}{\ensuremath{\mathcal{R}}}

\newcommand{\VIdx}[1]{V_{\Nat}^{#1}}
\newcommand{\unif}{\ensuremath{\stackrel{{\scriptscriptstyle ?}}{=}}}
\newcommand{\unifl}{\ensuremath{\stackrel{{\scriptscriptstyle \circlearrowleft}}{=}}}

\newcommand{\shidx}[1]{\ensuremath{S^{#1}}}
\newcommand{\exd}[2]{\ensuremath{\mathbf{ex}_{#1}^{#2}}}

\def\QEQ{{%
    \setbox0\hbox{{\large $\circlearrowleft$}}%
    \rlap{\hbox to \wd0{\hss \unif \hss}}\box0
}}

\makeatletter
\usepackage{scalerel}
\newcommand*{\shifttext}[2]{%
  \settowidth{\@tempdima}{#2}%
  \makebox[\@tempdima]{\hspace*{#1}#2}%
}
\newcommand\fat[1]{\ThisStyle{\ooalign{%
  \kern.46pt$\SavedStyle#1$\cr\kern.33pt$\SavedStyle#1$\cr%
  \kern.2pt$\SavedStyle#1$\cr$\SavedStyle#1$}}}

\makeatother

\newcommand{\lyus}[1]{
	\if#10 \ensuremath{\mbox{\includegraphics[scale=.025]{yus}}}
	 \else
	  \if#11 \ensuremath{\mbox{\includegraphics[scale=.035]{yus}}}
	   \else
	    \if#12 \ensuremath{\mbox{\includegraphics[scale=.045]{yus}}}
	     \else
	      \if#13 \ensuremath{\mbox{\includegraphics[scale=.055]{yus}}}
	       \else
	        \if#14 \ensuremath{\mbox{\includegraphics[scale=.065]{yus}}}
	         \else
	          \if#15 \ensuremath{\mbox{\includegraphics[scale=.075]{yus}}}
	          \else \ensuremath{\mbox{\includegraphics[scale=.045]{yus}}}
	          \fi 
	         \fi 
	        \fi 
	       \fi 
	      \fi 
	     \fi 
	     \makebox[.05cm][l]{}
}

\newcommand{\byus}[1]{
	\if#10 \ensuremath{\mbox{\includegraphics[scale=.025]{byus}}}
	 \else
	  \if#11 \ensuremath{\mbox{\includegraphics[scale=.035]{byus}}}
	   \else
	    \if#12 \ensuremath{\mbox{\includegraphics[scale=.045]{byus}}}
	     \else
	      \if#13 \ensuremath{\mbox{\includegraphics[scale=.055]{byus}}}
	       \else
	        \if#14 \ensuremath{\mbox{\includegraphics[scale=.065]{byus}}}
	         \else
	          \if#15 \ensuremath{\mbox{\includegraphics[scale=.075]{byus}}}
	          \else \ensuremath{\mbox{\includegraphics[scale=.045]{byus}}}
	          \fi 
	         \fi 
	        \fi 
	       \fi 
	      \fi 
	     \fi 
\makebox[.05cm][l]{}
}

\newcommand{\ibyus}[1]{
	\if#10 \ensuremath{\mbox{\includegraphics[scale=.025]{ibyus}}}
	 \else
	  \if#11 \ensuremath{\mbox{\includegraphics[scale=.035]{ibyus}}}
	   \else
	    \if#12 \ensuremath{\mbox{\includegraphics[scale=.045]{ibyus}}}
	     \else
	      \if#13 \ensuremath{\mbox{\includegraphics[scale=.055]{ibyus}}}
	       \else
	        \if#14 \ensuremath{\mbox{\includegraphics[scale=.065]{ibyus}}}
	         \else
	          \if#15 \ensuremath{\mbox{\includegraphics[scale=.075]{ibyus}}}
	          \else \ensuremath{\mbox{\includegraphics[scale=.045]{ibyus}}}
	          \fi 
	         \fi 
	        \fi 
	       \fi 
	      \fi 
	     \fi 
\makebox[.05cm][l]{}
}

\newcommand{\byusa}[1]{
	\if#10 \ensuremath{\mbox{\includegraphics[scale=.025]{byusarrow}}}
	 \else
	  \if#11 \ensuremath{\mbox{\includegraphics[scale=.035]{byusarrow}}}
	   \else
	    \if#12 \ensuremath{\mbox{\includegraphics[scale=.045]{byusarrow}}}
	     \else
	      \if#13 \ensuremath{\mbox{\includegraphics[scale=.055]{byusarrow}}}
	       \else
	        \if#14 \ensuremath{\mbox{\includegraphics[scale=.065]{byusarrow}}}
	         \else
	          \if#15 \ensuremath{\mbox{\includegraphics[scale=.075]{byusarrow}}}
	          \else \ensuremath{\mbox{\includegraphics[scale=.045]{byusarrow}}}
	          \fi 
	         \fi 
	        \fi 
	       \fi 
	      \fi 
	     \fi 
}

\newcommand{\qmarrow}[1]{
	\if#10 \ensuremath{\mbox{\includegraphics[scale=.065]{qarrow}}}
	 \else
	  \if#11 \ensuremath{\mbox{\includegraphics[scale=.075]{qarrow}}}
	   \else
	    \if#12 \ensuremath{\mbox{\includegraphics[scale=.085]{qarrow}}}
	     \else
	      \if#13 \ensuremath{\mbox{\includegraphics[scale=.095]{qarrow}}}
	       \else
	        \if#14 \ensuremath{\mbox{\includegraphics[scale=.105]{qarrow}}}
	         \else
	          \if#15 \ensuremath{\mbox{\includegraphics[scale=.075]{qarrow}}}
	          \else \ensuremath{\mbox{\includegraphics[scale=.045]{qarrow}}}
	          \fi 
	         \fi 
	        \fi 
	       \fi 
	      \fi 
	     \fi 
}

\newtheorem{proposition}{Proposition}

\IEEEoverridecommandlockouts
\begin{document}

\title{A Special Case of Schematic Syntactic Unification}

\author{David M. Cerna$^{1}$\IEEEauthorblockN{\thanks{Supported by the Linz Institute of Technology (LIT) Math$_{LP}$ project (LIT-
2019-7-YOU-213) funded by the state of upper Austria.}}
\IEEEauthorblockA{ Czech Academy of Sciences Institute of Computer Science (CAS ICS),
  Prague, Czechia\\ Research Institute for Symbolic Computation (RISC), Johannes Kepler University, Linz, Austria\\
dcerna@\{cs.cas.cz,risc.jku.at\}}

}


%Telephone: (800) 555--1212, Fax: (888) 555--1212}


\maketitle

\begin{abstract}
  We present a unification problem based on first-order syntactic unification which ask whether every problem in a schematically-defined sequence of unification problems is unifiable, so called {\em loop unification}. Alternatively, our problem may be formulated as a recursive procedure calling first-order syntactic unification on certain bindings occurring in the solved form resulting from unification.  Loop unification is closely related to \textit{Narrowing} as the schematic constructions can be seen as a rewrite rule applied during unification, and \textit{primal grammars}, as we deal with recursive term constructions. However, loop unification relaxes the restrictions put on variables  as fresh as well as used extra variables may be introduced by rewriting. In this work we consider an important special case, so called {\em semiloop  unification}. We provide a sufficient condition for unifiability of the entire sequence based on the structure of a sufficiently long initial segment. It remains an open question whether this condition is also necessary for semiloop unification and how it may be extended to loop unification. 
\end{abstract}


% \fi
\IEEEpeerreviewmaketitle

\section{Introduction}
Methods of cut-elimination based on resolution~\cite{DBLP:journals/jsc/BaazL00} allow the use of techniques from automated reasoning for formal proof transformation and analysis. When a formal proof uses inductive reasoning, proof analysis is best performed using a schematic representation of the formal proof. When performing resolution based analysis on schematically represented formal proofs, concepts such as unification must also be generalized to so called {\em unification schema}. We focus on a particular type of unification schema and its associated unification problem.    

Schematic unification, as discussed in this paper, originated from computational proof analysis using methods of cut-elimination based on resolution theorem proving~\cite{DBLP:journals/tcs/BaazHLRS08,DBLP:journals/jsc/BaazL00}. While it may seem natural to handle proof analysis in the presence of induction using a higher-order formulation, a schematic first-order formulation turned out to be more fruitful and resulted in extraction of an elementary proof from F\"{u}rstenberg's proof of the infinitude of primes~\cite{DBLP:journals/tcs/BaazHLRS08}. Earlier work in this direction concerning schematic representation was carried out using formalisms built for a specific problem. This approach is not easily generalized to analysis of a wide-range of proofs.

 In~\cite{DBLP:conf/tbillc/DunchevLRW13}, the first general purpose formalism was presented, however, the associated unification problem was not formally specified. Manual construction of the appropriate substitutions was expected~\cite{DBLP:conf/cade/CernaL16}. An improved formalism was later developed which relied on schematic automated reasoning methods for construction of the unifiers~\cite{DBLP:journals/logcom/LeitschPW17}. Note that, while these unifiers resulted from schematic constructions, they themselves were not schematically defined.
 
The automated reasoning methods for induction used in~\cite{DBLP:journals/logcom/LeitschPW17} are quite powerful in comparison to other existing methods~\cite{DBLP:books/daglib/0013358}, however, the overall method is quite weak proof analytically. This observation lead to the recent generalization of the proof analysis method introduced in~\cite{DBLP:journals/jar/CernaLL21}. This new framework introduces the formal definition of the schematic unification which motivated this  work.
 
Concerning the schematic theorem prover used in~\cite{DBLP:journals/logcom/LeitschPW17}, it represents another line of research, involving schematic representations~\cite{DBLP:conf/cade/AravantinosCP10a,DBLP:conf/time/AravantinosCP11,DBLP:journals/fuin/AravantinosEP13}. In~\cite{DBLP:journals/fuin/AravantinosEP13}, inductive theorem proving over schematic clause sets is considered, however, schematic unification, as we discuss it, was not.

The unification problem presented in this work is closely related to {\em equational narrowing}~\cite{ESCOBAR2009103} and {\em primal grammars}~\cite{HERMANN1997111}. Concerning narrowing, we can consider our schematic construction as rewrite rules that are applicable to particular constants occurring in the unification problem. However, our schematic construction allows the rewrite rules to introduce {\em extra variables} and some of these extra variables may be present in the unification problem and not necessarily present in the term upon which narrowing is being applied. This automatically implies that we would need to consider narrowing over a conditional rewrite system~\cite{DBLP:journals/scp/CholewaEM15}, however, we nonetheless break the freshness condition required by narrowing as ``used'' variables may be introduced.

Similarly for primal grammars, restricting the variables occurring on the right-hand side of rewrite rules is essential to the decidability of its unification problem. Our schematic formalization breaks the variable occurrence restriction and thus cannot be expressed as a primal grammar. 
While this implies that neither formalism is adequate for representing our unification problem, the close relationship to our work implies that research in these areas may benefit from improved understanding of schematic unification.

In this paper, we introduce a variation of the schematic unification problem discussed in~\cite{DBLP:journals/jar/CernaLL21} and an important special case when only one side of the problem is schematic, what we refer to as {\em semiloop unification}. This seems like a significant simplification, though the problem remains difficult as  complex cyclic structures may occur during  unification. While we are not able to provide a decision procedure for semiloop unification, we are able to provide a sufficient condition for unifiability and conjecture the necessity of this condition. 
\section{Preliminaries}
Let $\Var$ be a countably infinite set of variable symbols. A {\em variable class} is a pair $(Z,<)$, where $Z\subset \Var$ (countably infinite) and $<$ is a strict well-founded total linear order. Associated with each variable class $(Z,<)$ is the {\em successor function $Suc_{<}^Z(\cdot)$ of the class} which has the following properties: 
\begin{itemize}
\item if $x\in Z$, then $x < Suc_{<}^Z(x)$, and 
\item if $x,y\in Z$ and $Suc_{<}^Z(x) < Suc_{<}^Z(y)$ then $x < y$. 
\end{itemize}
 when it is clear from context, we write $Suc(\cdot)$ for $Suc_{<}^Z(\cdot)$.
To simplify notation we will consider classes of the form $\VIdx{x} = ( \{ x_i\ \vert\ i\in \Nat\}, <_{\Nat})$ where $<_{\Nat}$ is the strict well-founded total linear order of the natural numbers and $x\not \in \{ x_i\ \vert\ i\in \Nat\}$. The successor function associated with $\VIdx{x}$ is defined as $Suc(x_i) = x_{i+1}$. Note that this successor function satisfies the above properties when $x_0$ is the minimal element with respect to the ordering of $\VIdx{x}$. Let $x_i\in \VIdx{x}$, then $|x_i|=i$. Unless otherwise stated, we assume the classes $\VIdx{x}$ and $\VIdx{y}$ contain distinct variables when  $x\not = y$.

In addition to a standard first-order term signature $\Sigma$ we require a countably infinite set  of recursion variables $\RecV$. Recursion variables will be denoted using $\hat{\cdot}$.  By $\mathcal{T}(\Sigma,S,Z)$ we denote the term algebra whose members are constructed using the signature $\Sigma$, $Z = (\bigcup_{x\in X} \VIdx{x})$ where $X\subset \mathcal{V}$, and $S\subset \RecV$. For $t\in \mathcal{T}(\Sigma,S,Z)$, $\mathit{var}(t)$ denotes the set of variables occurring in $t$.

Given a term $t\in \mathcal{T}(\Sigma,S,Z)$, for $Z = (\bigcup_{x\in X} \VIdx{x})$, we can generate the {\em successor term of $t$ modulo $Z$},  by applying the {\em shift operator} which is defined recursively as follows: 
\begin{itemize}
\item $\shidx{Z}(f(t_1,\cdots, t_n)) = f(\shidx{Z}(t_1),\cdots, \shidx{Z}(t_n))$
\item for $z\in \VIdx{x}$, s.t. $x\in X$ ,  $\shidx{Z}(z) = Suc(z)$ 
\item for $\hat{a}\in S$, $\shidx{Z}(\hat{a}) =\hat{a}$
\end{itemize}
When $Z$ is clear from context we will simply write $\shidx{}(\cdot)$.

Substitutions are finite sets of pairs $\{ z_1\mapsto t_1 ,\cdots  , z_1\mapsto t_n \}$, were  $z$’s are pairwise distinct variables or recursion variables. They represent functions that map a finite number of $z_i$ to $t_i$ and any other variable to itself.  The notions of substitution domain and range are also
standard and are denoted, respectively, by $\mathit{Dom}$ and $\mathit{Ran}$. Substitutions are denoted by lower case Greek letters, while the identity substitution is denoted by $\mathit{Id}$.
We use postfix notation for substitution applications, writing $t\sigma$ instead of $\sigma(t)$. As usual, the application $t\sigma$ affects only the free occurrences of variables from $\mathit{Dom}(\sigma)$ in $t$. 
The composition of $\sigma$ and $\theta$ is written as juxtaposition $\sigma\theta$ and is defined as $x(\sigma\theta) = (x\sigma)\theta$ for all $x$. Another standard operation, restriction of a substitution $\sigma$ to a set of variables $S$, is denoted
by $\sigma\vert_S$ .

A {\em Unification problem}, denoted by $\{ t_1\unif s_1, \cdots, t_n\unif s_n\}$, is a set of pairs of terms. A unification problem $S$ is said to be in irreducible form if for every pair $t_i\unif s_i$ one of the following holds: the head symbol of $t_i$ and $s_i$ mismatch, or $t_i$ is a variable which does not occur in any other pairs of the unification problem, or $t_i$ is a variable occurring as a proper subterm of $s_i$. If a unifier can be extracted from the irreducible form than we refer to the irreducible form as a {\em solved form}. 
For more details concerning these and related concepts connected with first-order syntactic unification such as  {\em most gernal unifier (m.g.u)}, we refer the reader to~\cite{unif}.

\section{Loops, Semiloops, and Unification}
What separates loop unification from syntactic unification is that certain terms can be extended and thus produce new unification problems.  
\begin{definition}
\label{extendableAndfixed}
let $Z\subset \Var$, $s\in \mathcal{T}(\Sigma,\{\hat{a}\},(\bigcup_{z\in Z} \VIdx{z}))$, and $t\in  \mathcal{T}(\Sigma,\emptyset,(\bigcup_{z\in Z} \VIdx{z}))$. Then we refer to $s$ as {\em $(\Sigma,Z,\hat{a})$-extendable} and $t$ as {\em $(\Sigma,Z)$-fixed}. Associated with each $(\Sigma,Z,\hat{a})$-extendable term $s$ is an unary operator $\exd{\hat{a}}{s}(\cdot)$  defined  recursively as follows: 
\begin{itemize}
\item $\exd{\hat{a}}{s}(f(t_1,\cdots, t_n)) = f(\exd{\hat{a}}{s}(t_1),\cdots, \exd{\hat{a}}{s}(t_n))$
\item for $z\in Z$ and $i\in \mathbb{N}$, $\exd{\hat{a}}{s}(z_i) = z_i$
\item  $\exd{\hat{a}}{s}(\hat{a}) =s$ 
\item for $\hat{b}\in \mathcal{R}$ such that $\hat{a}\not = \hat{b}$, $\exd{\hat{a}}{s}(\hat{b}) =\hat{b}$ 
\end{itemize}
\end{definition}
\noindent Both $\shidx{}(\cdot)$ (shift operator) and $\exd{\hat{a}}{s}(\cdot)$ (extend operator) may be applied to substitutions as follows:

\begin{itemize}
\item $\shidx{}(\sigma)=\{  \shidx{}(y) \mapsto \shidx{}(t) \ \vert \ y\sigma = t\}$
\item $\exd{\hat{a}}{s}(\sigma)=\{ y \mapsto \exd{\hat{a}}{s}(t) \ \vert \  y\sigma = t \}$
\end{itemize}

\begin{definition}
\label{loops}
Let $s$ be $(\Sigma,Z,\hat{a})$-extendable, $t$ be $(\Sigma,Z,\hat{b})$-extendable. Then we refer to the pair  $\langle s,t \rangle$  as a {\em $(\Sigma,Z,\hat{a},\hat{b})$-loop}. When it is clear from context we will write loop for $(\Sigma,Z,\hat{a},\hat{b})$-loop. 
\end{definition}

Note that it is not necessary for the terms of a loop to be extendable. The important subclass we discuss in this paper, semiloops, considers the case when only one of the loop terms is extendable. 

\begin{definition}
\label{semiloops}
Let $s$ be $(\Sigma,Z,\hat{a})$-extendable and $r$ $(\Sigma,Z)$-fixed. Then we refer to the pair $\langle s,r \vert$ ($\vert r,s\rangle$)  as  a {\em left (right) $(\Sigma,Z,\hat{a})$-semiloop}. When it is clear from context we will write semiloop for $(\Sigma,Z,\hat{a})$-semiloop. 
\end{definition}

Left and right semiloops are defined symmetrically, and thus results concerning  left semiloops will hold for right semiloops. For the rest of this work, when discussing semiloops, we will exclusively consider left semiloops unless we state otherwise.

\begin{proposition}
Every $(\Sigma,Z,\hat{a})$-semiloop is a $(\Sigma,Z,\hat{a},\hat{b})$-loop.
\end{proposition}
The above proposition will help us avoid repeating definitions which apply to both loops and semiloops. 

\begin{definition}
\label{extensions}
Let $\langle s,t\rangle$ be a  $(\Sigma,Z,\hat{a},\hat{b})$-loop and $n\in \Nat$. The  $n$-extension of $\langle s,t\rangle$, denoted by $\langle s,t\rangle_n$, may be defined recursively as follows: 
\begin{itemize}
\item $\langle s,t\rangle_0 = \langle \hat{a},\hat{b}\rangle$
\item $\langle s,t\rangle_{n+1} = \langle \exd{\hat{a}}{s}(\shidx{}(s')),\exd{\hat{b}}{t}(\shidx{}(t'))\rangle$ where $\langle s,t\rangle_{n} = \langle s',t'\rangle$.
\end{itemize}
\end{definition}
In the case of semiloops, application of the operators may be ignored for the fixed term, i.e.  $\langle s,t\vert_0 = \langle \hat{a},t\vert$, and $\langle s,t\vert_{n+1} = \langle \exd{\hat{a}}{s}(\shidx{}(s')),t\vert$ where $\langle s,t\vert_{n} = \langle s',t\vert$. This convention stops unnecessary shifting of variables in fixed terms, and thus, simplifies definitions and lemmas discussed later in this paper. 

\begin{proposition}
Every extension of a loop is a loop. 
\end{proposition}
This observation implies the following: If $\langle s,t\rangle_{k} = \langle s',t'\rangle$, then there exists operators, $\exd{\hat{a}}{s'}(\cdot)$ and $\exd{\hat{b}}{t'}(\cdot)$, for $\hat{a}$ occurring in $s'$ and $\hat{b}$ occuring in $t'$ . This will be essential for proving certain results. 
\begin{example}
The following are  $(\{h\},\{\VIdx{x},\VIdx{y}\},\hat{a},\hat{b})$-loops: 
\begin{itemize}
\item[a)] $\langle h(h(x_1,h(x_1,h(x_1,x_1))),\hat{a})\ , \ h(h(h(y_{1},y_{2}),y_{4}),\hat{b})\rangle$
\item[b)] $\langle h(h(x_6,h(x_1,x_6)),\hat{a})\ ,\ h(\hat{b},h(y_1,h(y_2,y_1)))\rangle$
\item[c)] $\langle h(h(x_6,h(\hat{a},x_6)),x_1)\ ,\ h(y_1,h(\hat{b},h(y_2,y_1)))\rangle$
\item[d)] $\langle h(h(x_6,h(x_1,x_6)),\hat{a})\ ,\ h(y_1,h(y_2,,y_1))\vert$
\item[e)] $\langle h(h(x_6,h(x_1,x_6)),x_4)\ ,\ h(y_1,h(y_2,y_1))\rangle$
\end{itemize}
Example $d)$ is a $(\{h\},\{\VIdx{x},\VIdx{y}\},\hat{a})$-semiloop and example $e)$ is a pair of first-order terms. Note that example $e)$ may be considered
a  $(\{h\},\{\VIdx{x},\VIdx{y}\},\hat{a},\hat{b})$-loop or a $(\{h\},\{\VIdx{x},\VIdx{y}\},\hat{a})$-semiloop. In either case, it is a degenerate example. All extensions of 
example $e)$ are equivalent by the convention mentioned earlier (we do not apply operators to fixed terms).

\end{example}

\begin{example}
Consider the following Semiloop $\langle s,t\vert =\langle h(h(x_6,h(x_1,x_6)),\hat{a})\ ,\ h(y_1,h(y_2,y_1))\vert$. We define $t(n) = h(x_{n+6},h(x_{n+1} ,x_{n+6}))$. It's first few extensions are as follows:
\begin{itemize}
\item $\langle s,t\vert_0 = \langle \hat{a}\ ,\ h(y_1,h(y_2,y_1)) \vert$
\item $\langle s,t\vert_1 = \langle h(h(x_6,h(x_1,x_6)), \hat{a}),\ h(y_1,h(y_2,y_1))\vert$
\item $\langle s,t\vert_2 = \langle h(t(1),h(t(0),\hat{a})),\ h(y_1,h(y_2,y_1)) \vert$
\item $\langle s,t\vert_3 = \langle (h(t(2),h(t(1),h(t(0),\hat{a}))),\ h(y_1,h(y_2,y_1)) \vert$
\end{itemize}
\end{example}

Recursion variables allow the construction of loops extensions through substitution by a predefined term. Concerning unification, we will treat recursion variables as a special type of variable with the following restrictions: Let $\langle s,t\rangle$ be a loop and $S$ the solved form of $s\unif t$. Let $\sigma$ be an m.g.u constructed from $S$, then we assume for any $\hat{a}\in \RecV$, $\hat{a}\sigma =t$ iff  $t\not \in \Var$. The idea behind this restriction is that variables have a higher precedence than recursion variables with respect to unification. This restriction will be helpful when defining a sufficient condition for unifiability. 
\begin{definition}
Let $\langle s,t\rangle$ be a $(\Sigma,Z,\hat{a},\hat{b})$-loop such that  solving $s\unif t$ results in an m.g.u. $\sigma$. We refer to $\langle s,t\rangle$  as {\em extendably unifiable} if  $\mathit{Dom}(\sigma)\cap \{\hat{a},\hat{b}\} \not = \emptyset$ .
\end{definition}
Observe that extendably unifiable implies that an unaddressed unification problem exists in the solved form resulting from the unification problem $s\unif t$. These unaddressed problems may be introduced by extending $s$, $t$ or both. 
\begin{example}
Consider the following semiloop, $$\langle h(h(x_6,h(x_1,x_6)),\hat{a})\ ,\ h(y_1,h(y_2,y_1))\vert.$$ 
The solved form of the induced unification problem  $h(h(x_6,h(x_1,x_6)),\hat{a})\unif h(y_1,h(y_2,y_1))$ contains:
$$\{ y_1\unif h(x_6,h(x_1,x_6)) \ , \ \hat{a}\unif h(y_2,h(x_6,h(x_1,x_6)))\}$$
and thus is extendably unifiable. However, a slight variation of the loop, namely $\langle h(h(x_6,h(x_1,x_6)),\hat{a})\ ,\ h(y_1,y_2)\vert$, is not extendably unifiable as its solved form contains 
$$\{ y_1\mapsto h(x_6,h(x_1,x_6)) \ , \ y_2 \mapsto \hat{a}\}.$$
\end{example}

\begin{definition}[Loop Unification Problem]
Let $X\subseteq Z$. Given a $(\Sigma,Z,\hat{a},\hat{b})$-loop $\langle s,t\rangle$ such that  $s$ is $(\Sigma,(Z\setminus X),\hat{a})$-extendable and $t$ is $(\Sigma,X ,\hat{b})$-extendable, the  {\em loop unification problem}, denoted by $s\unifl t$, is the problem of deciding if every extension of $\langle s,t\rangle$ is unifiable. We refer to such loops as {\em loop unifiable}. 
\end{definition}
\begin{example}
Let us consider the semiloop: 
$\langle s, t \vert =$ $$\langle h(h(h(x_2,h(x_1,x_1)),x_3),\hat{a}), \ h(h(y_4,y_3),h(y_1,y_2)) \vert$$
Given that the variables of $s$ and $t$ are disjoint  we can define the loop unification problem $s\unifl t$. $\langle s, t \vert_0$ is always unifiable so we can ignore it. $\langle s, \ t \vert_1$ has the following unifier: 
$$\{ y_3\mapsto x_3 \ , \ y_4 \mapsto h(x_2,h(x_1,x_1)) \ , \ \hat{a}\mapsto h(y_1,y_2) \}.$$
Thus, $\langle s,  t \vert_1$ is extendably unifiable. What about $\langle s,  t \vert_2$?
$$ \langle  h(t(1),h(t(0),\hat{a}))\ , \ h(h(y_4,y_3),h(y_1,y_2)) \vert$$
where $t(n) = h(h(x_{n+2},h(x_{n+1},x_{n+1})),x_{n+3})$. It has the following unifier:
$$ \begin{array}{l} 
\{ y_3\mapsto x_4 ,\   y_4 \mapsto h(x_3,h(x_2,x_2)),\\  y_1 \mapsto h(h(x_2,h(x_1,x_1)),x_3),\ y_2 \mapsto \hat{a}  \}
\end{array}$$
Notice that this extension is not extendably unifiable. From this unifier we can build a unifier for every extension greater than 2 (See Theorem~\ref{finiteunif}) and thus, this example is loop unifiable. 
\end{example}

\begin{example}
Let us consider the semiloop: $\langle s,t\vert = $
$$\langle h(h(h(x_2,x_1),h(x_2,x_3)),\hat{a})\ , \ h(h(y_3,y_1),h(y_4,y_4)) \vert$$
Once again the variables of $s$ and $t$ are disjoint and we can define the loop unification problem $s\unifl t$. $\langle s,t\vert_1$ has the following unifier: 
$$\{ y_3\mapsto h(x_2,x_1) \ , \ y_1 \mapsto h(x_2,x_3) \ , \ \hat{a}\mapsto h(y_4,y_4) \}.$$
Thus, $\langle s,t\vert_1$ is extendably unifiable. What about the $\langle s,t\vert_2$?
$$\langle h(t(1),h(t(0),\hat{a}))\ , \ h(h(y_3,y_1),h(y_4,y_4))\vert$$
where $t(n) = h(h(x_{n+2},x_{n+1}),h(x_{n+2},x_{n+3}))$. It has the following unifier:
$$\begin{array}{l}\{ y_3\mapsto h(x_3,x_2), \  y_4 \mapsto h(h(x_2,x_1),h(x_2,x_3)), \\ 
 y_1 \mapsto h(x_3,x_4), \ \hat{a} \mapsto h(h(x_2,x_1),h(x_2,x_3)) \}.\end{array}$$
Thus $\langle s,t\vert_2$ is extendably unifiable. Now consider $\langle s,t\vert_3:$
$$\langle h(t(2),h(t(1),h(t(0),\hat{a})))\ , \ h(h(y_3,y_1),h(y_4,y_4))\vert$$
Notice that the irreducible form contains an occurrence check on the variable $x_2$ and thus $\langle s,t\vert_3$ is not unifiable. 
$$\begin{array}{l} \{ y_3\unif h(x_4,x_3), \ y_4 \unif h(h(x_3,x_2),h(x_3,x_4)), \ \\  \ y_1 \unif h(x_4,x_5), \  \hat{a} \unif  h(x_3,x_4), \ x_3 \unif h(x_2,x_1), \\ \ x_2 \unif h(x_2,x_3) \}.\end{array}$$
\end{example}
In both of the above examples only a finite number of extensions are extendably unifiable. However, this need not be the case. We can distinguish types of loop unifiability based on the number of extensions which are extendably unifiable. 

\begin{definition}
Let the loop $\langle s,t\rangle$ of $s \unifl t$ be loop unifiable. We say  $\langle s,t\rangle$ is  {\em infinitely loop unifiable} if for every $n\in \mathbb{N}$, $\langle s,t\rangle_n$ is {\em extendably unifiable}. Otherwise, we say $\langle s,t\rangle$  is finitely  loop unifiable.
\end{definition}
\begin{example}
The following is a simple example of a semiloop which is infinitely loop unifiable:
$$\langle s,t\vert = \langle h(h(x_1,x_1),\hat{a})\ , \ h(y_1,y_1)\vert$$
Notice that the solved form of  every extension  will contain either $\hat{a} \unif h(y_1,y_1)$ or $\hat{a}\unif h(x_1,x_1)$. Thus, every extension is extendably unifiable. 
\end{example}
Interestingly, even somewhat simple examples can lead to non-trivial infinite loop unifiability:
\begin{example}
Consider the semiloop 
$ \langle s, t\vert =$ $$\langle h(\hat{a},t(0)), \  h(h(h(h(y_1,y_1),y_1),y_1),y_1) \vert$$
where $t(n) = h(h(h(x_{n+1},x_{n+1}),x_{n+1}),x_{n+1})$ 
While this may not be clear from simple observation, for $n\geq 1$, the solved form of  $\langle s,t\vert_{3n}$ will contain  
$\hat{a} \unif h(h(t(1),t(1)),t(1)),$
the solved form of  $\langle s,t\vert_{3n+1}$ will contain  $\hat{a} \unif h(t(1),t(1)),$
the solved form of  $\langle s,t\vert_{3n+2}$ will contain  $\hat{a} \unif t(1).$
This pattern repeats for all $m$-extensions where $m>2$. 
\end{example}
In the following section we provide a  sufficient condition for finite unifiability of semiloops. Note that we only consider the case when the semiloop is constructed from two variable classes, one for the left term and one for the right term. This simple case is already non-trivial and provides interesting results which would be obfuscated by the combinatorial complexity of handling more classes. We leave the handling of a finite number of variable classes as well as full loop unification to future work. 

\section{A Sufficient Condition for Finite Unifiability of Semiloops}
In this section we provide a sufficient condition for finite unifiability of semiloops: Given a semiloop $\langle s,t\vert$, if $\langle s,t\vert_n$ satisfies the condition and for all $0\leq j\leq n$,  $\langle s,t\vert_j$ is unifiable, then we can build a unifier for any extension. For the remainder of this section we will consider $\langle s,t\vert$ to be a $(\sigma,\{x,y\} ,\hat{a})$-semiloop. The semiloop of $\langle s,t\vert_k$ will be denoted by $\langle s_k,t\vert$, and, if unifiable, its m.g.u. will be denoted by $\sigma_k$. We assume variables are not renamed in the process of constructing $\sigma_k$. Furthermore, we assume that $\mathit{var}(s) \subset \VIdx{x}\cup \{\hat{a}\}$ and $\mathit{var}(t) \subset \VIdx{y}$. Under these conditions we can construct the loop unification problem $s\unifl t$. 
\begin{lemma}
\label{shiftunif}
If $\langle s_k,t\vert$ is unifiable, then $\shidx{}(s_k)\unif t$ is unifiable. 
\end{lemma}
\begin{proof}
$\shidx{}(\cdot)$ is essentially a variable renaming, thus,  $\shidx{}(\sigma_k)$ is the m.g.u. 
\end{proof}

\begin{lemma}
\label{exunif}
If $\langle s_k,t\rangle$ is unifiable and the following conditions hold:
\begin{itemize}
\item $\hat{a}\not \in \mathit{Dom}(\sigma_k)$
\item for all $z\in  \mathit{Dom}(\sigma_k)$ s. t.  $\hat{a}\in \mathit{var}(z\sigma_k)$, $z\not \in \mathit{var}(s\sigma_k)$.
\end{itemize}
then $\exd{\hat{a}}{s}(s_k)\unif t$ is unifiable  by the m.g.u $\{ z\mapsto \exd{\hat{a}}{s\sigma_k}(r)\ \vert\ z\sigma_k = r\}$. 
\end{lemma}
\begin{proof}
This follows from the fact that $z$ does not occur in $\exd{\hat{a}}{s\sigma_k}(r)$ and the structure of the term at every position, except positions at which $\hat{a}$ occurs, is the same in $\exd{\hat{a}}{s}(s_k)$ and $s_k$.
\end{proof}

\begin{corollary}
\label{shiftexunif}
If $\langle s_k,t\rangle$ is unifiable and the following conditions hold:
\begin{itemize}
\item $\hat{a}\not \in \mathit{Dom}(\sigma_k)$
\item for all $z\in  \mathit{Dom}(\sigma_k)$ s.t. $\hat{a}\in \mathit{var}(z\sigma_k)$, $\shidx{}(z)\not \in \mathit{var}(s\theta)$ where  $\theta = \shidx{}(\sigma_k)$
\end{itemize}
then $\exd{\hat{a}}{s}(\shidx{}(s_k))\unif t$ is unifiable  by the m.g.u $\{ \shidx{}(z)\mapsto \exd{\hat{a}}{s\theta}(\shidx{}(r))\ \vert\  z\sigma_k= r\}$.
\end{corollary}
\begin{proof}
By Lemma~\ref{shiftunif}~\&~\ref{exunif}.
\end{proof}
\begin{corollary}
\label{missingx}
Let $k\geq 1$. If  $\langle s_k,t\vert$ is unifiable, but not extendably unifiable, then $\langle s_{k+1},t\vert$ is not extendably unifiable. 
\end{corollary}
\begin{proof}
Follows from Corollary~\ref{shiftexunif} and the definition of $n$-extension of $\langle s,t\rangle$.
\end{proof}

While Corollary~\ref{missingx} tells us that once an extension is not extendably unifiable all larger extension cannot be extendably unifiable, it does not tell us whether or not all larger extension are actually unifiable.

\begin{example}
\label{exampleProblem}
Consider the following example:
$$\langle s,t\vert = \langle h(x_2,h(x_4,\hat{a})), h(y_1,y_1)\vert$$
The unifier of $\langle s,t\vert_1$ is as follows:
$$\{ y_1\mapsto h(x_4,\hat{a}) \ , \  x_2\mapsto h(x_4,\hat{a})\}$$
By Corollary~\ref{shiftexunif} the unifier of $\langle s,t\vert_2$ is as follows:
$$\begin{array}{c} \{ y_1\mapsto h(x_5,h(x_2,h(x_4,\hat{a}))), \\ 
 \ x_3\mapsto h(x_5, h(x_2,h(x_4,\hat{a})))\}\end{array}$$
Once again, by Corollary~\ref{shiftexunif} the unifier of $\langle s,t\vert_3$ would be:
$$\begin{array}{c} \{ y_1\mapsto h(x_6,h(x_3,h(x_5,h(x_2,h(x_4,\hat{a}))))), \\ \  x_4\mapsto h(x_6, h(x_3,h(x_5,h(x_2,h(x_4,\hat{a})))))\}\end{array}$$
however, there is an occurrence check in the irreducible form. The problem is that $x_2$ is smaller than the largest variable occurring in the extendable term of the loop and shifting may cause an occurrence check.
\end{example}
We need to add an extra constraint to Corollary~\ref{shiftexunif} in order to extend unifiability to larger extensions.  
\begin{theorem}
\label{finiteunif}
Let $k>0$. Then if $\langle s_k,t\vert$ is unifiable and the following conditions hold:
\begin{itemize}
\item $\hat{a}\not \in \mathit{Dom}(\sigma_k)$
\item for all $z\in  \mathit{Dom}(\sigma_k)$ s.t. $\hat{a}\in \mathit{var}(z\sigma_k)$, $|\shidx{}(z)| >m$, where $m = \max_{x\in \mathit{var}(s\theta)} |x|$
and  $\theta = \shidx{}(\sigma_k)$.
\end{itemize}
Then for all $j\geq k$, $\langle s_j,t\vert$ is unifiable and thus the loop unification problem $s\unifl t$ is unifiable. 
\end{theorem}
\begin{proof}
Proof by induction on $j$ using Corollary~\ref{shiftexunif}~\&~\ref{missingx}.
\end{proof}

Theorem~\ref{finiteunif} tells us that if for a semiloop $\langle s,t\vert$ we find an extension which is not extendably unifiable we only need to test finitely many extensions to decide if $\langle s,t\vert$ is finitely unifiable. By ensuring $z$ is larger than the largest variable introduced through extending the semiloop we ensure that $z$ cannot be in occurrence check with an introduced variable. Most of the work which follows concerns showing that a similar condition exists for infinite unifiability.
 \section{A Sufficient Condition for Infinite Unifiability}
In this section we provide a useful decomposition for extendably unifiable extensions of semiloops. As in the previous section, we will only consider semiloops constructed using two distinct variables classes. The motivation for this restriction is similar to motivation provided at the beginning of the previous section.  For the remainder of this section we will consider $\langle s,t\vert$ to be a $(\Sigma,\{x,y\} ,\hat{a})$-semiloop. The semiloop of $\langle s,t\vert_k$ will be denoted by $\langle s_k,t\vert$, and, if unifiable, its m.g.u. will be denoted by $\sigma_k$. We assume variables are not renamed in the process of constructing $\sigma_k$. Furthermore, we assume that $\mathit{var}(s) \subset \VIdx{x}\cup \{\hat{a}\}$ and $\mathit{var}(t) \subset \VIdx{y}$. Under these conditions we can construct the loop unification problem $s\unifl t$. 

\subsection{Decomposition of Extendably Unifiable Semiloops}
The following lemma motivates the decomposition introduced in Definition~\ref{decompOp}.
\begin{lemma}
\label{minidecompLemma}
Let $k\geq 1$. If  $\langle s_1,t\vert$, $\langle s_k,t\vert$, and $\langle s_{k+1},t\vert$ are extendably unifiable, then $\sigma_{k+1} = \mathit{sh}^k(\theta)\sigma$ where: 

\begin{itemize}
\item $\mathit{sh}^k(\cdot) = \overbrace{\shidx{}(\cdots \shidx{}(\cdot)\cdots)}^k$
\item $\sigma_1 = \theta\{ \hat{a}\mapsto t'\}$ where  $ \hat{a}\not \in \mathit{Dom}(\theta)$
\item $\sigma$ is the m.g.u of $s_{k}\mu\unif \mathit{sh}^k(t')$ where $\mu = \mathit{sh}^k(\theta)$.  
\end{itemize} 
\end{lemma}
\begin{proof}
By definition,  $\langle s_{k+1},t\vert$ is the following semiloop: 

$$\langle \overbrace{\exd{\hat{a}}{s}(\shidx{}(\cdots \exd{\hat{a}}{s}(\shidx{}(s))\cdots))}^k,t\vert$$
and $\langle s_{k},t\vert$ is the following semiloop:
$$\langle \overbrace{\exd{\hat{a}}{s}(\shidx{}(\cdots \exd{\hat{a}}{s}(\shidx{}(s))\cdots))}^{k-1},t\vert$$
Because $s_k$ is an extendable term, we have the operator $\exd{\hat{a}}{s_k}(\cdot)$ and may rewrite  $\langle s_{k+1},t\vert$ as follows: 

$$\langle \exd{\hat{a}}{s_k}(\overbrace{\shidx{}(\cdots \shidx{}(s))\cdots)}^k),t\vert$$

or as $\langle \exd{\hat{a}}{s_k}(s'),t\rangle$ where by $s'=\mathit{sh}^k(s)$. Note that the only difference between $s'\unif t$ and $s\unif t$ is the shift of the variables names (see Lemma~\ref{shiftunif}). Thus, if $\theta\{ \hat{a}\mapsto t'\}$ is a unifier of $s\unif t$ then $\mathit{sh}^k(\theta)\{\hat{a}\mapsto \mathit{sh}^k(t')\}$  is a unifier of $s'\unif t$. Furthermore, $\exd{\hat{a}}{s_k}(\hat{a})\mathit{sh}^k(\theta)\unif \mathit{sh}^k(t')$ is a unification problem which results from replacing $\hat{a}$ in $s'\unif t$  by $s_k$. This problem must have a unifier $\sigma$ by our assumptions and is equivalent to $s_k\mathit{sh}^k(\theta)\unif \mathit{sh}^k(t')$ . Now concatenating the two unifiers we get  $\mathit{sh}^k(\theta)\sigma$  as the m.g.u of  $\langle s_{k+1},t\vert$.
\end{proof}

Using Lemma~\ref{minidecompLemma} we can construct a decomposition of $\sigma_k$  using all $\sigma_j$ for $1	\leq j< k$. The idea behind this decomposition is that the unification problem $s_k\unif t$ associated with the semiloop $\langle s,t\vert$ can be decomposed into two unification problems, namely $\mathit{sh}^k(s)\unif t$ and $s\theta\unif t'$ where $\theta$ and $t'$ are derived from the solved form of $\mathit{sh}^k(s)\unif t$. If we iteratively perform this decomposition we can construct the operator introduced in the following Definition.

\begin{definition}
\label{decompOp}
Let $k\geq 1$. If for all  $0\leq j\leq k$, $\langle s_j,t\vert $ is extendably unifiable, then {\em the decomposition operator} $\mathit{D}(\sigma,s,t,k)$  is defined as follows:

$$\mathit{D}(\sigma,s,t,n+1) \equiv \mathit{sh}^n(\theta) \mathit{D}(\shidx{V}(\sigma\theta),s,\shidx{V}(t'),n) $$
$$\mathit{D}(\sigma,s,t,0) \equiv \{ \hat{a}\mapsto t\} $$
where  $\theta\{ \hat{a}\mapsto t'\}$ is the m.g.u. of $s\sigma\unif t$,$\hat{a}\not \in \mathit{Dom}(\theta)$, and $V = \VIdx{x}$.
\end{definition}

\begin{lemma}
\label{decompLemma}
Let $k\geq 1$. If for all  $0\leq j\leq k$, $\langle s_j,t\vert $ is extendably unifiable, then $\sigma_k = \mathit{D}(\mathit{Id},s,t,k)$.
\end{lemma}
\begin{proof}
This statement is provable by induction on $k$ and application of Lemma~\ref{minidecompLemma}. Observe that we only shift the variables which occur in $s$. This restriction is required so that our decomposition conforms with the definition of extensions of semiloops  (see Definition~\ref{extensions}).
\end{proof}

\begin{example}
Consider the following semiloop: 
$$\langle s, t\vert = \langle h(h(x_6,h(x_1,x_6)),\hat{a})\ , \ h(y_1,h(y_2,y_1))\vert$$
We will use the abbreviation $t(n) = h(x_{n+6},h(x_{n+1},x_{n+6}))$ below. 
$\langle s, t\vert_3$ has the following form:
$$\langle h(t(2),h(t(1),h(t(0),\hat{a})))\ , \ h(y_1,h(y_2,y_1))\rangle$$
 
The solved form of $s_3\unif t$ is as follows:
$$\begin{array}{c}
\{ y_1  \unif  h(x_8,h(x_3,x_8)), \
y_2  \unif  h(x_7,h(x_2,x_7))\\
x_8  \unif  h(x_6,h(x_1,x_6)), \
\hat{a}  \unif  h(x_3,h(x_6,h(x_1,x_6))\}
\end{array}$$
The unifier $\sigma_3$ constructable from this solved form can be written as follows: 
$$
D(\mathit{Id},h(h(x_6,h(x_1,x_6)),\hat{a}),h(y_1,h(y_2,y_1)),3) = $$ $$
\mathit{sh}^2(\theta^2)D(\mathit{sh}^1(\theta^2),s,h(y_2,t(1))),2)=$$ $$
\mathit{sh}^2(\theta^2)\mathit{sh}^1(\theta^1)D(\mathit{sh}^1(\theta^1),s,t(2),1)=$$ $$
\mathit{sh}^2(\theta^2)\mathit{sh}^1(\theta^1)\theta^0D(\mathit{sh}^1(\theta^0),s,h(x_4,t(1)),0)=$$ $$
\mathit{sh}^2(\theta^2)\mathit{sh}^1(\theta^1)\theta^0\{\hat{a} \mapsto h(x_3,h(h(x_6,h(x_1,x_6))\}
$$

which is equivalent to the unifier: 

$$\{y_1  \mapsto  h(x_8,h(x_3,x_8))\ , \ y_2  \mapsto  h(x_7,h(x_2,x_7))\ , \ $$ $$ x_8  \mapsto  h(x_6,h(x_1,x_6))\ , \ \hat{a} \mapsto h(x_3,h(h(x_6,h(x_1,x_6))\}$$
where
$$\theta^2 = \{y_1  \mapsto  h(x_6,h(x_1,x_6))\}$$
$$\theta^1 = \{y_2  \mapsto  h(x_6,h(x_1,x_6))\}$$
$$\theta^0 = \{x_8  \mapsto  h(x_6,h(x_1,x_6))\}$$

Note that, surprisingly this loop is not infinitely unifiable as the 14-extension is not unifiable. This is due to a problem similar to what was pointed out in Example~\ref{exampleProblem}. When formulating our sufficient condition we will need to consider large enough extensions to avoid such occurrence checks.  
\end{example}
Notice that Lemma~\ref{decompLemma}   requires  all previous substitutions $\theta$ to compute the next substitution. However, for large $k$, only a few variables will occur in both $\langle s,t\vert_1$  and $\langle s,t\vert_k$. Thus, we can further restrict $\sigma$ to a set of ``small'' variables. We do this by computing the size of the interval between the smallest and largest variable occurring in $s$. Once we reach an extension larger than this the number of occurrences of the smallest variable in $s$ is maximized. Let $\Delta(s)\in\Nat$ be defined as follows:

$$\Delta(s) = (\max_{z\in\mathit{var}(s)}  |z| - \min_{z\in\mathit{var}(s)} |z|)$$
Notice, if $r = \min_{z\in\mathit{var}(s)} |z|$, then $r+\Delta(s) = \max_{z\in\mathit{var}(s)}  |z|$, and thus $\forall j\geq 1 (x_{r+\Delta(s)+j}\not \in \mathit{var}(s))$.

\begin{definition}
\label{ExtendeddecompOp}
Let $k\geq 1$.  If for all  $0\leq j\leq k$, $\langle s_j,t\vert $ is extendably unifiable, then the {\em extended decomposition operator} $\mathit{D}'(\sigma,s,t,k)$ is defined as follows:

$$\mathit{D}'(\sigma,s,t,n+1) \equiv \mathit{sh}^k(\theta) \mathit{D}'(\sigma^\Delta ,s,\shidx{V}(t'),n) $$
$$\mathit{D}'(\sigma,s,t,0) \equiv \{ \hat{a}\mapsto t\} $$
where  $\theta\{ \hat{a}\mapsto t'\}$ is the m.g.u. of $s\sigma\unif t$, $\hat{a}\not \in \mathit{Dom}(\theta)$, $m =\min_{z\in\mathit{var}(s)} |z|$, $V=\VIdx{x}$, and $\sigma^\Delta$ is the substitution which coincides with $\shidx{V}(\sigma\theta)$ over $\{x_i\vert i\in [m,m+\Delta(s)]\}$ and is otherwise identical to the identity substitution.
\end{definition}

\begin{lemma}
\label{decompLemma2}
Let $k\geq 1$.  If for all  $0\leq j\leq k$, $\langle s_j,t\vert $ is extendably unifiable, then $\sigma_k = \mathit{D}'(\mathit{Id},s,t,k)$.
\end{lemma}
\begin{proof}
Follows from Lemma~\ref{decompLemma} by restricting the domain of $\shidx{V}(\sigma\theta)$ to variables occurring in the range $[m,m+\Delta(s)]$ where $m =\min_{z\in\mathit{var}(s)} |z|$. Variables which are outside this range (indexed by a natural number larger than $m+\Delta(s)$)  do not occur in $s$ or any extension of $s$ as the substitutions are shifted during decomposition as well.
\end{proof}
\subsection{The Sufficient Condition}
This decomposition is almost enough for us to provide a sufficient condition for a loop to be infinitely unifiable. However, as mentioned earlier, we need one additional condition based on $\Delta(s)$. If the variable indices form an interval, then the condition mentioned below is not necessary. 
\begin{lemma}
There exists a semiloop $\langle s,t\vert$ and $1 \leq k< \Delta(s)$ such that the  $x_j\in\mathit{var}(s_k)$ and  $x_j\not \in \mathit{var}(s)$, for $r< j\leq r+\Delta(s)$. 
\end{lemma}
\begin{proof}
Consider $\langle h(h(x_6,h(x_1,x_6)),\hat{a}), h(y_1,h(y_2,y_1))\vert.$
\end{proof}
While this may not seem to be overtly problematic, it implies that the value of $\sigma^\Delta$ and $t'$ may change at the r$^{th}$ step of the decomposition depending on which extension we are considering, i.e. some variables in the interval  $[m,m+\Delta(s)]$ may only occur in higher extension. To avoid this issue we need to consider extensions which are large enough.

\begin{lemma}
\label{fixedvariables}
Let $k,j\geq 2\cdot\Delta(s)$ such that $j\geq k$. Then for every $l\in   [m,m+\Delta(s)]$ where $m =\min_{x_i\in\mathit{var}(s)} |x_i|$, the set of positions at which $x_l$ occurs in  $\langle s,t\vert_k$  is the same as the set of positions at which $x_l$ occurs in $\langle s,t\vert_j$. 
\end{lemma}
\begin{proof}
Let $j = k+r$. We can prove this statement by induction on $r$ and the observation that $s_k$ is a subterm of $s_j$. By definition, none of the variables in $\mathit{var}(s_j)\setminus \mathit{var}(s_k)$ are in the set $\{x_l\ \vert \ l\in   [m,m+\Delta(s)]\}$. All the variables in $\mathit{var}(s_j)\setminus \mathit{var}(s_k)$ are produced by applying $\shidx{}$ at least $2\cdot\Delta(s)+1$ times and thus cannot occur in $\{x_l\ \vert \ l\in   [m,m+\Delta(s)]\}$.
\end{proof}
Lemma~\ref{fixedvariables} can be strengthened by observing that the number of variables whose occurrences are fixed relative to $\langle s,t\vert_k$ in $\langle s,t\vert_j$ is a function of $k$. As $k$ increases the number of variables whose maximum number of occurrences is fixed increases. Though, this observation is not essential for proving our sufficient condition. 

What is needed to prove our sufficient condition is a secondary decomposition which allows one to take the unifier of a extension which is extendably unifiable and decompose it into ``sub''-substitutions. These substitutions, as we will show, can be composed to construct a unifier for other extensions. Under certain conditions, the substitutions can construct a unifier for any extension. However, Unlike our previous sufficient condition, all extensions of the considered semiloops are extendably unifiable. 
\begin{definition}
\label{Segmentation}
Let $r\geq 0$, $0\leq j\leq i\leq r$, and  $\sigma_{r+1}$ decomposes as follows:

$$\mathit{D}'(\mathit{Id},s,t,r+1) = \Theta(r+1)\mathit{D}'(\sigma^\Delta_1,s,t_1,r)$$
$$\vdots$$
$$\mathit{D}'(\sigma_{r-i+1},s,t_{r-i+1},i) = \Theta(i)\mathit{D}'(\sigma^\Delta_{r-i},s,t_{r-i},i-1) $$
$$\vdots$$
$$\mathit{D}'(\sigma_{r-j+1},s,t_{r-j+1},j) = \Theta(j)\mathit{D}'(\sigma^\Delta_{r-j},s,t_{r-j},j-1) $$
$$\vdots$$
$$\mathit{D}'(\sigma_{r+1},s,t_{r+1},0) =  \{ \hat{a}\mapsto t_{r+1}\}$$ 

where $\Theta(n) = \mathit{sh}^{n-1}(\theta_{n-r})$. then the {\em $(i,j)$-segment of $\sigma_{r+1}$} is 
$$ \mathit{Seg}(\sigma_{r+1},i,j) = \Theta(i)\mathit{Seg}(\sigma_{r+1},i-1,j)$$
$$ \mathit{Seg}(\sigma_{r+1},j,j) = \Theta(j)$$
Furthermore, the {\em $(i,j)$-substitution of $\sigma_{r+1}$} is 
$$ \mathit{Sub}(\sigma_{r+1},i,j) = \Theta(i)\mathit{Sub}(\sigma_{r+1},i-1,j)$$
$$ \mathit{Sub}(\sigma_{r+1},j,j) = \Theta(i)\{ \hat{a}\mapsto t_{r-j}\}$$
\end{definition}

Observe that Definition~\ref{Segmentation} uses the decomposition of Lemma~\ref{decompLemma2} to construct the segments and substitutions.  
\begin{lemma}
\label{seqshift}
Let $r,k\geq 0$ such that $\sigma_r$ and $\sigma_k$ exists and are decomposible using Lemma~\ref{decompLemma2}, and $0\leq j\leq r\leq k$. Then  $$\mathit{sh}^{k-r}(\mathit{Seg}(\sigma_{r},r,j)) = \mathit{Seg}(\sigma_{k},k,j+(k-r)).$$
\end{lemma}
\begin{proof}
This follows from the definition of an $(i,j)$-segment and the decomposition introduced in  Lemma~\ref{decompLemma2}.
\end{proof}
The sufficient condition for infinitely loop unifiability essentially states that if we are considering a large enough extension (above $2\cdot\Delta(s)$), then when a cycle is discovered in the decomposition we can use this cycle to build a unifier for any extension. The bound of $2\cdot\Delta(s)$ is essential as the result is dependent on the property of extensions introduced in Lemma~\ref{fixedvariables}.
\begin{theorem}
\label{sufficientCondition}
Let $r>2\cdot\Delta(s)$. If for all $0\leq j\leq r+1$, $\langle s,t\vert_j$ is extendably unifiable and $\sigma_{r+1}$ decomposes  as follows 
$$\mathit{D}'(\mathit{Id},s,t,r+1) = \Theta(r+1) \mathit{D}'(\sigma^\Delta_1,s,t_1,r)$$
$$\vdots$$
$$\mathit{D}'(\sigma_{r-i+1},s,t_{r-i+1},i) = \Theta(i)\mathit{D}'(\sigma^\Delta_{r-i},s,t_{r-i},i-1) $$
$$\vdots$$
$$\mathit{D}'(\sigma_{r-j+1},s,t_{r-j+1},j) = \Theta(j)\mathit{D}'(\sigma^\Delta_{r-j},s,t_{r-j},j-1) $$

$$\vdots$$
$$\mathit{D}'(\sigma_{r+1},s,t_{r+1},0) =  \{ \hat{a}\mapsto t_{r+1}\}, $$

where $\Theta(n) = \mathit{sh}^{n-1}(\theta_{n-r})$ and $j\leq i\leq 2\cdot\Delta(s)$, then if  $t_{r-i} =t_{r-j}$ and $\sigma^\Delta_{r-i} = \sigma^\Delta_{r-j}$ then  $\langle s,t\vert$ is infinitely loop unifiable.
\end{theorem}
\begin{proof}
We can prove this statement by showing that from the unifier of $\langle s,t\vert_r$ we can construct a unifier for  $\langle s,t\vert_k$ where $k\geq r$. What is important to observe is that the cycle present in the decomposition will repeat modulo $i-j$. We can use this fact to construct a sequence of segments which together construct $\sigma_k$. When $k= r$, the proof is trivial as we already have a unifier. when $k>r$ we need to consider the $(r,i+1)$-segment $\mathit{Seg}(\sigma_r,r,i+1)$. By Lemma~\ref{seqshift}, we can construct the   $\mathit{Seg}(\sigma_{k},k,i+(k-r)+2)$ segment of $\sigma_k$. However, note that this segment is followed by the segment $\mathit{Seg}(\sigma_{k},i+(k-r)+1,j+(k-r)+2)$ which is equivalent to $\mathit{sh}^{k-r}(\mathit{Seg}(\sigma_{r},i,j+1))$ by Lemma~\ref{seqshift}. If $k-r+1$ is divisible by $i-j$, then for some $\alpha$ we get the following:

$$\sigma_k = \mathit{sh}^{\gamma}(\mathit{Seg}(\sigma_r,r,i+1))\mathit{sh}^{\alpha(i-j)}(\mathit{Seg}(\sigma_{r},i,j+1))$$
$$\mathit{sh}^{(\alpha-1)(i-j)}(\mathit{Seg}(\sigma_{r},i,j+1))\cdots \mathit{Sub}(\sigma_{r},i,j)$$
where $\gamma = k-r+1$. Otherwise, $k-r+1 = \alpha(i-j)+ \beta$. In this case we get 
$$\sigma_k = \mathit{sh}^{\gamma}(\mathit{Seg}(\sigma_r,r,i+1))\mathit{sh}^{\alpha(i-j)+\beta}(\mathit{Seg}(\sigma_{r},i,j+1))$$
$$\mathit{sh}^{(\alpha-1)(i-j)+\beta}(\mathit{Seg}(\sigma_{r},i,j+1))\cdots \mathit{sh}^{\beta}(\mathit{Seg}(\sigma_{r},i,j+1))$$ $$\mathit{Sub}(\sigma_{r},i,i-\beta)$$
where $\gamma = k-r+1$. Thus, we can construct a unifier for all extensions of the semiloop $\langle s,t\vert$.
\end{proof}
\begin{example}
Consider the semiloop 
$$\langle h(\hat{a},h(h(x_1,x_1),x_1))\ , \  h(h(h(y_1,y_1),y_1),y_1) \vert,$$
which we abbreviate as $\langle s,t\vert$ and we define $t(n) = h(x_{n+1},x_{n+1}),x_{n+1})$ . Note that $\Delta = 0$. Now consider the $D'$ decomposition of $\langle s,t\vert_5$:
$$D'(\mathit{Id},s, t,5) =\mathit{sh}^4(\theta^4)D'(\mathit{Id},s,h(h(t(1),t(1)),t(1)),4)$$
When decomposing $D'(\mathit{Id},s,h(h(t(1),t(1)),t(1)),4)$ we end up with the first step 
of the cycle as presented in Theorem~\ref{sufficientCondition}. Thus, $i=3$ based on the following decomposition:
$$\mathit{sh}^4(\theta^4)\mathit{sh}^3(\theta^3)D'(\mathbf{Id},s,\mathbf{h(t(1)),t(1)))},\mathbf{3})$$
Decomposing $D'(Id,s,h(t(1)),t(1))),3)$ does not yet result in a cycle:
$$\mathit{sh}^4(\theta^4)\mathit{sh}^3(\theta^3)\mathit{sh}^2(\theta^2)D'(\mathit{Id},s,t(1),2)
$$
Decomposing $D'(\mathit{Id},s,t(1),2)$ results in the second step 
of the cycle as presented in Theorem~\ref{sufficientCondition} and thus, $j=2$. 
$$\mathit{sh}^4(\theta^4)\mathit{sh}^3(\theta^3)\mathit{sh}^2(\theta^2)\mathit{sh}^1(\theta^1)D(\mathbf{Id},s,\mathbf{h(t(1),t(1))},\mathbf{1})$$
At this point we know, by Theorem~\ref{sufficientCondition}, that this semiloop is infinitely loop unifiable. The rest of the decomposition is as follows: 
$$\mathit{sh}^4(\theta^4)\mathit{sh}^3(\theta^3)\mathit{sh}^2(\theta^2)\mathit{sh}^1(\theta^1)\theta^0\{\hat{a} \mapsto  h(h(x_2,x_2),x_2)\}$$
Where the substitutions $\theta^i$ are as follows:
\begin{align*}
\theta^4 =& \{ y_2 \mapsto h(h(x_1,x_1),x_1) \}\\
\theta^3 =& \{ x_2 \mapsto x_1 \}\\
\theta^2=& \{  x_2 \mapsto x_1 \}\\
\theta^1 =& \{ x_2 \mapsto h(h(x_1,x_1),x_1) \}\\
\theta^0=& \{  x_2 \mapsto x_1 \}
\end{align*}
\end{example}

To illustrate why the restriction to extensions greater than $2\cdot\Delta(2)$ is necessary consider the following example where a cycle occurs, but not every extension is unifiable. 

\begin{example}
Let $\langle s,t\vert$ where: 
$$s= h(h(x_{1},h(x_{16},h(x_{32},h(x_{1},h(x_{16},x_{32}))))),\hat{f})$$ 
$$t= h(y_{1},h(y_{2},h(y_{3},h(y_{1},h(y_{2},y_{3})))))$$

The solved form of  $\langle s,t\vert_{11}$ results in a unifier $\sigma_{11}$ which when decomposed using Lemma~\ref{decompLemma2} contains the follow steps
$$D'(\mathit{Id},s,h(x_{4},h(x_{19},h(x_{35},h(x_{4},h(x_{19},x_{35}))))),9)$$
$$D'(\mathit{Id},s,h(x_{4},h(x_{19},h(x_{35},h(x_{4},h(x_{19},x_{35}))))),4).$$

This fits the cycle requirement of the decomposition outlined in Lemma~\ref{sufficientCondition}. Yet,  $\langle s,t\rangle_{28}$  is not unifiable. Interestingly, the loop $\langle s,t\vert$ where 
$$s= h(h(x_{1},h(x_{16},h(x_{31},h(x_{1},h(x_{16},x_{31}))))),\hat{a})$$ 
$$t= h(y_{1},h(y_{2},h(y_{3},h(y_{1},h(y_{2},y_{3})))))$$

is infinitely unifiable, though this is quite hard to show. The cycle outlined above still occurs, but stops at $\langle s,t\rangle_{17}$ where it is replaced by a more complex cyclic behavior. 
\end{example}
We conjecture that if a semiloop is infinitely unifiable then after a finite number of steps, a cycle, as presented in Theorem~\ref{sufficientCondition} will occur in the decomposition of an extension. Experimentally we have tested around 3 million different semiloops and have  failed to find a counterexample to this claim.

Together Theorem~\ref{finiteunif} and Theorem~\ref{sufficientCondition} provide a sufficient condition for semiloop unifiability. Unfortunately, developing a necessary condition for finite unifiability has so far remained as difficult as finding one for infinite unifiability. We leave further analysis to future work.

\section{Conclusion and Future Work}
In this paper we introduce the concept of {\em loop unification}, a variant of first-order syntactic unification which requires unifying an infinite sequence of recursively defined terms. Unlike first-order syntactic unification, loop unifiability has two facets, finite unifiability and infinite unifiability depending on how the terms in the infinite sequence depend on each other. We show that in both cases a sufficient condition for unifiability exists for a special subclass of loop unification, so called semiloop unification. Semiloop unification considers two infinite sequences, one which is recursively defined and the other for which every term is in the sequence is the same. Unfortunately, it is non-trivial to extend these conditions to be necessary as well, and thus this is left to future work. Also left to future work is extending our results to full loop unification and developing an algorithm based on our sufficient conditions. This algorithm can then be integrated into the computational proof analysis method introduced in~\cite{DBLP:journals/jar/CernaLL21}.

\bibliographystyle{IEEEtran}
\bibliography{references}
\end{document}